\documentclass[letterpaper,twocolumn]{article}
\usepackage{pbox}
\usepackage[raggedright]{titlesec}
\usepackage{color}
\usepackage{graphicx,enumerate}
\usepackage{amsfonts}
\usepackage{amsmath,pslatex}
\usepackage{url}
\usepackage{verbatim}
\graphicspath{ {images/} }
\usepackage{scrextend}
\usepackage{tikz}
\usepackage{graphics}
\usepackage{epsfig}
\usepackage{amsmath}
\usepackage{setspace}
\usepackage{algorithm}
\usepackage[noend]{algpseudocode}
\makeatletter
\def\BState{\State\hskip-\ALG@thistlm}
\newcommand{\species}{{\mathcal{S}}} 
\newcommand{\genes}{{\mathcal{G}}}
\newcommand{\baddups}{{\mathcal{D}}}
\newcommand{\proteins}{{\mathcal{P}}}
\newcommand{\leafset}{{\mathcal{L}}}
\newcommand{\vertices}{{\mathcal{V}}}
\newtheorem{corollary}{Corollary}
\newtheorem{lemma}{Lemma}

\newtheorem{theorem}{Theorem}
\newtheorem{proposition}{Proposition}

\newenvironment{proof}[1][Proof]{\textbf{#1.} }{\ \rule{0.5em}{0.5em}}
\newtheorem{remark}{Remark}

\date{}

\titlespacing{\section}{0pt}{*2.0}{*2.0}
\titlespacing{\subsection}{0pt}{*1.8}{*1.8}

\begin{document}

\title{\Large\textbf{Reconstructing Protein and Gene Phylogenies by Extending the Framework of Reconciliation}\normalsize}

\author{
	Esaie Kuitche$^*$, Manuel Lafond$^{\dag}$ and A\"ida Ouangraoua$^*$\\\\
	$^*$Department of Computer Science, Universit\'e de Sherbrooke\\  
	Sherbrooke, QC, J1K2R1, Canada\\
    (esaie.kuitche.kamela,aida.ouangraoua)@USherbrooke.ca
	\and
	\\    
	$^{\dag}$Department of Mathematics and Statistics, University of Ottawa\\
	Ottawa, ON,  K1N6N5, Canada\\
    mlafond2@UOttawa.ca
}
\maketitle 

\thispagestyle{empty}

\begin{center}
\large\textbf{Abstract}
\end{center}

The architecture of eukaryotic coding genes allows the production of several different protein isoforms by genes. Current gene phylogeny reconstruction methods make use of a single protein product per gene, ignoring information on alternative protein isoforms. These methods often lead to inaccurate gene tree reconstructions that require to be corrected before being used in phylogenetic tree reconciliation analyses or gene products phylogeny reconstructions. Here, we propose a new approach for the reconstruction of accurate gene trees and protein trees accounting for the production of alternative protein isoforms by the genes of a gene family. We extend the concept of reconciliation to protein trees, and we define a new reconciliation problem called \textsc{MinDRGT} that consists in finding a gene tree that minimizes a double reconciliation cost with a given protein tree and a given species tree. We define a second problem called \textsc{MinDRPGT} that consists in finding a protein tree and a gene tree minimizing a double reconciliation cost, given a species tree and a set of protein subtrees. We provide algorithmic exact and heuristic solutions for some versions of the problems, and we present the results of an application to the correction of gene trees from the Ensembl database.
An implementation of the heuristic method is available at https://github.com/UdeS-CoBIUS/Protein2GeneTree.
\vspace{2mm}

\medskip
\noindent
\textbf{keywords:} Protein Tree, Gene Tree, Species Tree, Reconciliation

\section{Introduction} 
\label{sec:intro}
Recent genome analyses have revealed the ability of eukaryotic coding genes to produce several transcripts and proteins isoforms.
This mechanism plays a major role in the functional diversification of genes \cite{keren2010,nilsen2010}. Still, current gene phylogeny reconstruction methods make use of a single protein product per gene that is usually the longest protein called the "reference protein", ignoring the production of alternative protein isoforms \cite{aakerborg2009,rasmussen2011,vilella2009}. It has been shown that these sequence-based methods often return incorrect gene trees \cite{hahn2007,rasmussen2011}. Thus, several methods have been proposed for the correction of gene trees
\cite{noutahi2016,wu2013}. 
Recently, a few models and algorithms aimed at  reconstructing the evolution of full sets of gene products along gene trees were introduced 
\cite{christinat2012,zambelli2010}. Some models have also been proposed to study the evolution of alternative splicing and gene exon-intron structures along gene trees \cite{irimia2007,keren2010}. All these models require the input of accurate gene trees and are biased when the input gene trees contain errors.  

Here, we explore a new approach in order to directly reconstruct accurate gene phylogenies and protein phylogenies while accounting for the production of alternative protein isoforms by genes. We introduce new models and algorithms for the reconstruction of gene phylogenies and full sets of proteins phylogenies using reconciliation \cite{eulenstein2010}. We present a model of protein evolution along a gene tree that involves two types of evolutionary event called \emph{protein creation} and \emph{protein loss}, in addition to the classical evolutionary events of speciation, gene duplication and gene loss considered in gene-species tree reconciliation. We propose an extension of the framework of gene-species tree reconciliation in order to define the concept of protein-gene tree reconciliation, and we introduce new reconciliation problems aimed at reconstructing optimal gene trees and proteins trees. First, we define the problem of finding a gene tree minimizing the sum of the protein-gene and gene-species reconciliation costs, given the protein tree and the species tree. We call this problem the \emph{Minimum Double Reconciliation Gene Tree} (\textsc{MinDRGT}) problem.  Second, we define the problem of jointly finding a protein tree and a gene tree minimizing the sum of the protein-gene and gene-species reconciliation costs, given the species tree and a set of subtrees of the protein tree to be found. We call this problem the \emph{Minimum Double Reconciliation Protein and Gene Tree} (\textsc{MinDRPGT}) problem.
 
In this paper, we first formally define, in Section \ref{sec:new_def}, the new protein evolutionary models and the related reconciliation problems, \textsc{MinDRGT} and \textsc{MinDRPGT}, for the reconstruction of gene phylogenies and full sets of protein phylogenies. In Section \ref{sec:MinDRGT_hardness}, we prove the NP-hardness of some versions of \textsc{MinDRGT}, especially the one called  \textsc{MinDRGT}$_{CD}$ that consists in minimizing the number of protein creation and gene duplication events. Next, in Section \ref{sec:MinDRGT_P=G}, we consider the \textsc{MinDRGT} problem in a special case where each gene is associated to a single protein. This restriction is relevant for the correction of gene trees output by sequence-based gene phylogeny reconstruction methods using a single protein per gene. Such methods make the unsupported assumption that each pair of leaf proteins in the protein tree is related through a least common ancestral node that corresponds to a speciation or a gene duplication event, and then, they output a gene tree equivalent to the reconstructed protein tree. In this perspective, the \textsc{MinDRGT} problem under the restriction that each gene is associated to a single protein,
allows pairs of proteins to be related through ancestral protein creation events, and then asks to find an optimal gene tree, possibly different from the input protein tree. In other terms, the protein tree is not confused with the gene tree, but it is used, together with the species tree, to guide the reconstruction of the gene tree.
We first show that, even with the restriction that each gene is associated to a single protein, for most versions the \textsc{MinDRGT} problem, the optimal gene tree may differ from the input protein tree.
We then exhibit a heuristic algorithm for the \textsc{MinDRGT} problem that consists in building the optimal gene tree by applying modifications on the input protein tree guided by the species tree. 

In Section \ref{sec:MinDRPGT}, we consider the \textsc{MinDRPGT} problem aimed at jointly reconstructing both a protein phylogeny and a gene phylogeny. We consider a restriction on the input data that requires the set of input protein subtrees to be the set of all inclusion-wise maximum subtrees of the target (real) protein tree $P$ that contain no protein creation node. Such an input can be obtained or approximated by using a soft-clustering approach in order to group proteins into clusters of orthologous proteins. Under this restriction, we present a polynomial exact algorithm for reducing \textsc{MinDRPGT} to a special case of \textsc{MinDRGT}
where the input protein tree $P$ is given with a partial labeling of its nodes. The algorithm consists in first reconstructing complete subtrees of $P$ and then combining them into $P$.

In Section \ref{sec:Appli}, the results of applying the heuristic algorithm for the 
correction of gene trees from the Ensembl database \cite{hubbard2002} 
show that the new framework allows to reconstruct gene trees whose 
double reconciliation costs are decreased, as compared to the initial Ensembl gene trees \cite{vilella2009}. 

 \section{Preliminaries: protein trees, gene trees and species trees}
\label{sec:prelim}

In this section, we introduce some preliminary notations:
$\species$ denotes a set of species, $\genes$ a set
of genes representing a gene family, and $\proteins$ a set of proteins
produced by the genes of the gene family. The three sets are accompanied with a
mapping function $s : \genes \rightarrow \species$ mapping each gene to its
corresponding species, and a mapping function
$g : \proteins \rightarrow \genes$ mapping each protein to
its corresponding gene. In the sequel, we assume that $\species$, $\genes$ and $\proteins$ satisfy $\{s(x) : x \in \genes \} = \species$  and
$\{g(x) : x \in \proteins \} = \genes$, without explicitly 
mentioning it.\\

\noindent {\bf Phylogenetic trees:} 
A tree $T$ for a set $L$ is a rooted binary tree whose leafset is $L$.
The leafset of a tree $T$ is denoted by $\leafset(T)$ and the set of node of $T$ is denoted by $\mathcal{V}(T)$. Given a node $x$
of $T$, the complete subtree of $T$ rooted at $x$ is denoted by $T[x]$. The
\emph{lowest common ancestor} (lca) in $T$ of a subset $L'$ of $\leafset(T)$,
denoted by $lca_T(L')$, is the ancestor common to all nodes in $L'$ that is
the most distant from the root of $T$. $T|_{L'}$ denotes the tree for
$L'$ obtained from the subtree $T'$ of $T$ induced by the subset of leaves $L'$ by removing all internal nodes of degree 2, except $lca_T(L')$, which is the root of $T'$.
Given an internal node $x$ of $T$,
the children of  $x$ are arbitrarily denoted by $x_l$ and $x_r$.\\

\noindent {\bf Proteins, genes, and species trees:} In the sequel, $S$ denotes a species tree for the set  $\mathcal{S}$, $G$ denotes a gene tree for the set $\mathcal{G}$, and  $P$ denotes a protein tree for the set $\mathcal{P}$.
The mapping function $s$  is extended to be defined from $\mathcal{V}(G)$ to $\mathcal{V}(S)$ such that if $x$ is an internal node of $G$, then $s(x) = lca_S(\{s(x'): x' \in \mathcal{L}(G[x])\})$, i.e. the image of a node $x\in \mathcal{V}(G)$ in $\mathcal{V}(S)$  is the lca in the tree $S$ of all the images of the leaves of $G[x]$ by $s$. 
Similarly, the mapping function  $g$ is extended to be defined from $\mathcal{V}(P)$ to $\mathcal{V}(G)$ such that if $x$ is an internal node of  $P$, then $g(x) = lca_G(\{g(x'): x' \in \mathcal{L}(P[x])\})$.\\

\noindent {\bf Gene-species tree reconciliation:} 
Each internal node of the species tree $S$ represents an ancestral species at the moment of a speciation
event (\emph{Spec}) in the evolutionary history of $\species$. 
The gene tree $G$ represents the evolutionary history of the genes of the gene family $\genes$, and
each internal node of $G$ represents an ancestral gene at the moment of a \emph{Spec} or a gene duplication event  (\emph{Dup}).

The \emph{LCA-reconciliation} of $G$ with $S$ is a labeling function $l_G$ from  $\vertices(G)-\leafset(G)$ to $\{Spec, Dup\}$ such that the label of
an internal node $x$ of $G$ is $l_G(x) = Spec$ if $s(x) \neq s(x_l)$ and
$s(x) \neq s(x_r)$, and $l_G(x) = Dup$ otherwise 
\cite{eulenstein2010,ma2000}. 
The LCA-reconciliation induces gene loss events on edges of $G$ as follows: given an edge $(x,y)$ of the tree $G$ such that $y = x_l$ or $y = x_r$, a gene loss event is induced on $(x,y)$ for each node located on the path between $s(x)$ and $s(y)$ in $S$ (excluding $s(x)$ and $s(y)$). If $l_G(x) = Dup$ and $s(x) \neq s(y)$, an additional loss event preceding all other loss events is induced on $(x,y)$ for $s(x)$.
Figure \ref{gene-species-reconciliation} presents a gene tree $G$ on a gene family $\mathcal{G} =\{a_2,a_3,b_0,b_1,b_2,b_3,c_1,c_2,c_3,d_3\}$ reconciled with a species tree $S$ on a set of species $\mathcal{S} =\{a,b,c,d\}$.

    \begin{figure}[!ht]
      \centering
      \includegraphics[width = 0.45\textwidth]{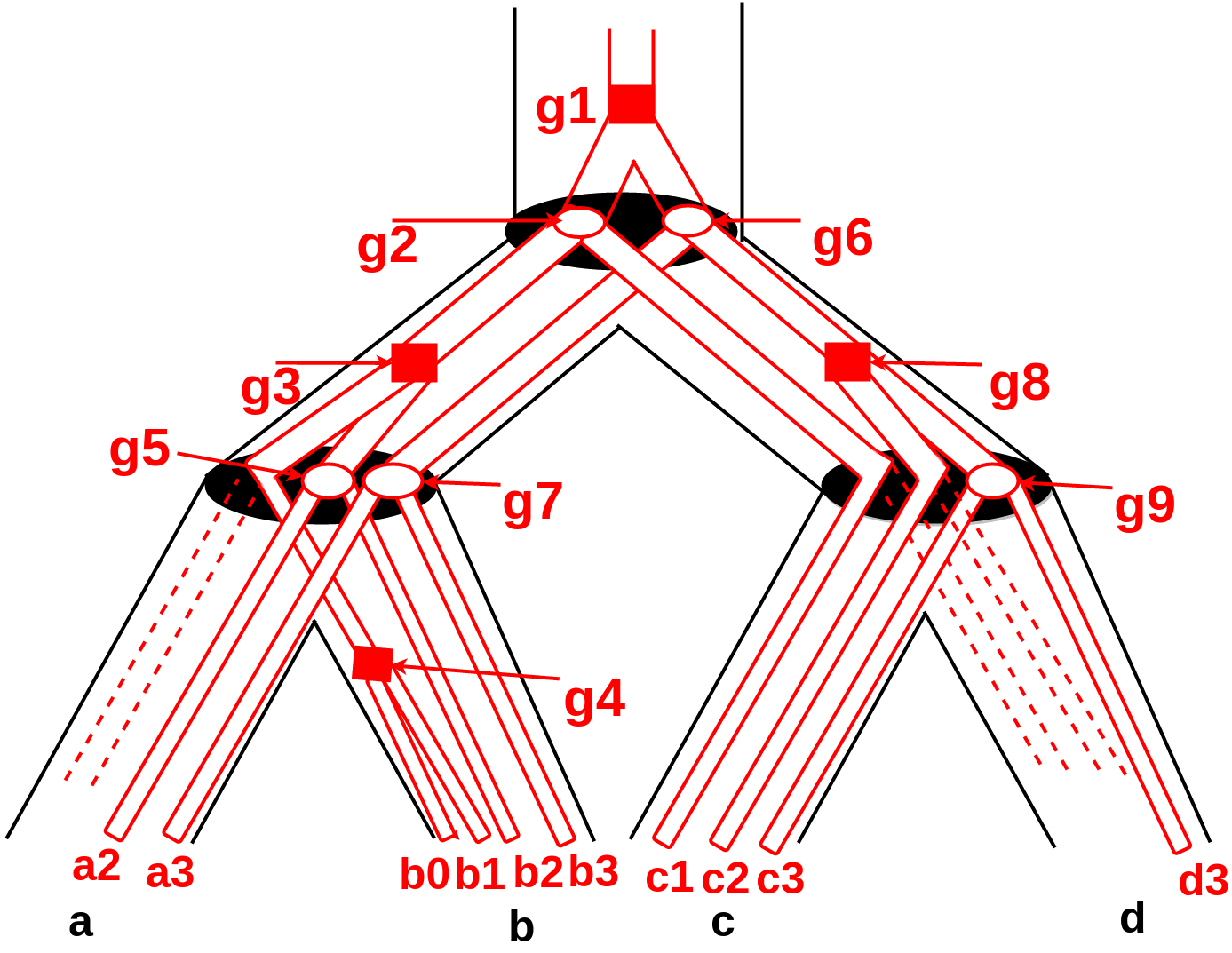}
      \caption{
      A gene tree $G$ on a gene family $\mathcal{G} =\{a_2,a_3,b_0,b_1,b_2,b_3,c_1,c_2,c_3,d_3\}$ such that
      $s(x_i) = x$ for any gene $x_i \in \genes$ and species $x \in \species=\{a,b,c,d\}$. The species tree $S$ is $((a,b),(c,d))$.
    $G$ is reconciled with $S$: a speciation node $x$ of $G$ is located inside the node $l_G(x)$ of $S$, and a duplication node $x$ is located on the edge $(p,l_G(x))$ of $S$ such that $p$ is the parent of $l_G(x)$ in $S$. The gene tree $G$ contains $9$ ancestral nodes: $g_1,g_3,g_4,g_8$ are duplications represented as squared nodes and $g_2,g_5,g_6,g_7,g_9$ are
      speciations represented as circular nodes. $G$ contains $3$ loss events whose locations are indicated with dashed edges. The same labeled gene tree $G$ is represented in Figure  \ref{protein-gene-reconciliation} (Top),  not embedded in $S$.}
      \label{gene-species-reconciliation}
    \end{figure}

The LCA-reconciliation $l_G$ induces the definition of three possible costs of reconciliation $C_{G\rightarrow S}$ between $G$ and $S$. The \emph{duplication cost} denoted by $D(G,S)$ is the number of nodes $x$ of $G$ such that $l_G(x) = Dup$.
The \emph{loss cost} denoted by $L(G,S)$  is the overall number of loss events induced by $l_G$ on edges of $G$.
The \emph{mutation cost} denoted by $M(G,S)$ is the sum of the duplication cost and the loss cost induced by $l_G$. In the example depicted in Figure \ref{gene-species-reconciliation}, the duplication cost is $4$, while the loss cost is $3$, and the 
mutation cost is $7$.\\

\vspace{-0.1cm}
\noindent {\bf Homology relations between genes:} Two genes $x$ and $y$ of the set $\mathcal{G}$ are called \emph{orthologs} if $l_G(lca_G(\{x,y\})) = Spec$, and \emph{paralogs} otherwise.  

\section{Model of protein evolution along a gene tree and problem statements} 
\label{sec:new_def} 
In this section, we first formally describe the new model of protein evolution along a gene tree. Next, we describe an extension of the framework of phylogenetic tree reconciliation that makes use of the new model, and we state new optimization problems related to the extended framework.

\noindent {\bf Protein evolutionary model:} 
The protein evolutionary model that we propose is based on the idea that the set of all proteins $\mathcal{P}$ produced by a gene family $\mathcal{G}$ have derived from a set $\mathcal{A}_P$ of common ancestral proteins that were produced by the ancestral gene located at the root of the gene tree $G$. This ancestral set of proteins evolved along the gene tree through different types of evolutionary and modification events including the classical events of speciation, gene duplication and gene loss. In the sequel, we consider that the ancestral set of proteins $\mathcal{A}_P$ is composed of a single ancestral protein that is the root of a tree for the set of proteins $\mathcal{P}$, but all definitions can be directly extended to protein forests, i.e sets of independent proteins trees rooted at multiple ancestral proteins.

A \emph{protein tree} $P$ is a tree for the set of proteins $\proteins$ representing the phylogeny of the proteins 
in $\proteins$.
Each internal node of $P$ represents an ancestral protein at the moment of a Spec, Dup, or a \emph{protein creation event} (\emph{Creat}). A protein creation event represents the appearance of a new protein isoform at a moment of the evolution of a gene family on an edge of the gene tree $G$. This evolutionary model is supported by recent studies on the evolution of gene alternative splicing patterns and inter-species comparison of gene exon-intron structures \cite{irimia2007,keren2010}. In particular, these studies have highlighted that alternative splicing patterns may be gene-specific or shared by groups of homologous genes \cite{barbosa2012,osorio2015}. A protein creation event thus leads to the observation of conserved protein isoforms called \emph{orthologous splicing isoforms} \cite{zambelli2010} in a group of homologous extant genes descending from the ancestral gene that underwent the protein creation event. Based on these observations, the present model of protein evolution allows to describe the evolution of the full set of proteins produced by a gene family along the gene tree of the family. Figure \ref{proteintree} presents an example of labeled protein tree for a set of proteins $\mathcal{P}=\{a21,a31, b01,b02, b11, b21, b31, c11, c12, c21, c31, d31\}$.\\

\begin{figure}[!ht]
        \centering
        \includegraphics[width = 0.45\textwidth]{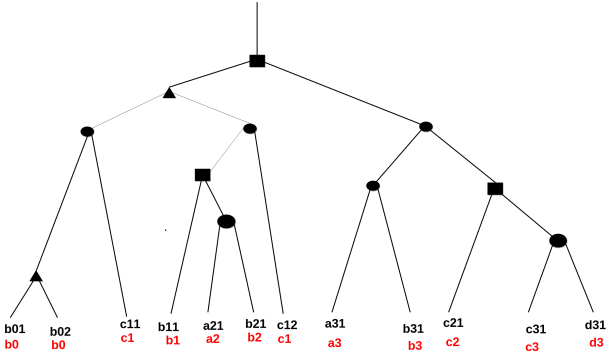}
        \caption{A protein tree $P$ on the set of proteins $\mathcal{P}=\{a21,a31, b01,b02, b11, b21, b31, c11, c12, c21, c31, d31\}$. The nodes of the tree are labeled as speciation (circular nodes), gene duplication (squared nodes), of protein creation events (triangular nodes). For each protein leaf $x_{ij}$ of $P$, the corresponding gene $x_i = g(x_{ij})$ is indicated below the protein. The LCA-reconciliation that resulted in the labeling of the nodes of $P$ is illustrated in Figure  \ref{protein-gene-reconciliation}.
}
        \label{proteintree}
\end{figure}

\noindent {\bf Protein-gene tree reconciliation:} 
We naturally extend the concept of reconciliation to protein trees as follows. The \emph{LCA-reconciliation} of $P$ with $G$ is a labeling function $l_P$
from $\vertices(P)-\leafset(P)$ to $\{Spec, Dup, Creat\}$ that labels an
internal node $x$ of $P$ as  $l_P(x) = Spec$ if $g(x) \neq g(x_l)$ and
$g(x) \neq g(x_r)$ and $l_G(g(x)) = Spec$, else $l_P(x) = Dup$ if
$g(x) \neq g(x_l)$ and $g(x) \neq g(x_r)$ and $l_G(g(x)) = Dup$, and
$l_G(x) = Creat$ otherwise. Note that, if $x$ 
is such that $\{g(y) | y \in \leafset(P[x_l])\} \cap \{g(y) | y \in \leafset(P[x_r])\} \neq \emptyset$, then 
$l_P(x) = Creat$, and $x$ is called an 
\emph{apparent creation} node.

    \begin{figure}[!ht]
      \centering
      \includegraphics[width = 0.42\textwidth]{gene.png}\\
      \includegraphics[width = 0.42\textwidth]{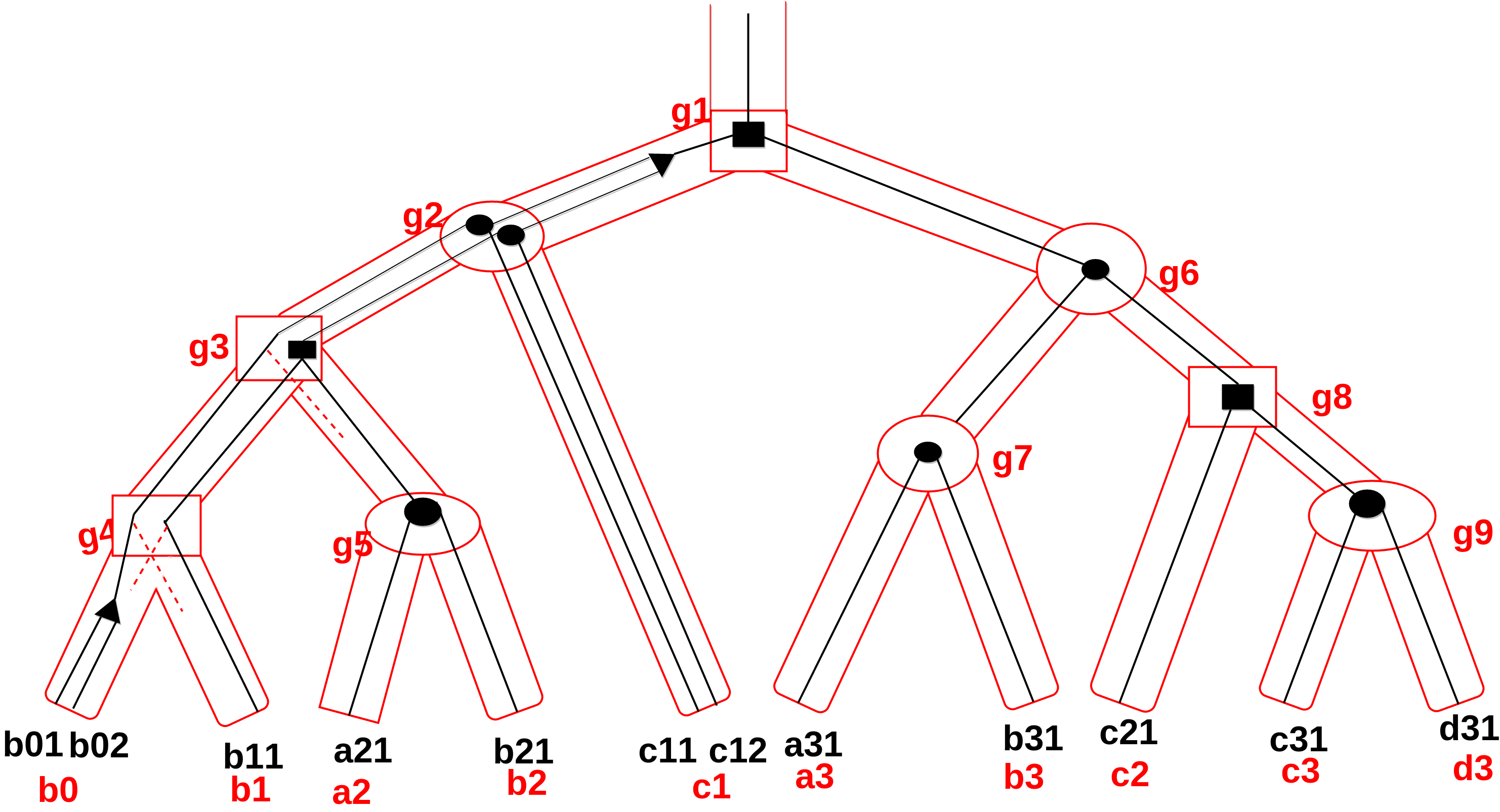}
      \caption{Top. The labeled gene tree $G$ of Figure \ref{gene-species-reconciliation}. Bottom. The protein tree $P$ of Figure \ref{proteintree}  reconciled with $G$.  For each internal node $x$ of $P$, the corresponding image $g(x)$ in $G$ is indicated. The protein tree $P$ contains $2$ protein creation nodes (triangular nodes), $3$ gene duplication nodes (squared nodes), $6$ speciation nodes (circular nodes), and $3$ protein loss events indicated as dashed lines. }
      \label{protein-gene-reconciliation}
    \end{figure}

Similarly to the LCA-reconciliation $l_G$, the LCA-reconciliation $l_P$ induces protein loss events on edges of $P$ as follows: given an edge $(x,y)$ of $P$ such that $y = x_l$ or $y = x_r$, a protein loss event is induced on $(x,y)$ for each node located on the path between $g(x)$ and $g(y)$ in $G$. If $l_P(x) = Creat$ and $g(x) \neq g(y)$, an additional protein loss event preceding all other protein loss events is induced on $(x,y)$ for $g(x)$.
A protein loss event corresponds to the loss of the ability to produce a
protein isoform for an ancestral gene at a moment of the evolution of the
gene family.

We define the following three costs of reconciliation $C_{P\rightarrow G}$
induced by 
the LCA-reconciliation $l_P$ of $P$ with $G$. The \emph{creation cost}
denoted  by $C(P,G)$ is the number of nodes $x$ of $P$ such that
$l_P(x) = Creat$.
The \emph{loss cost} denoted  by $L(P,G)$is the overall number of loss events
induced by $l_P$  on edges of $P$.
The \emph{mutation cost} denoted by $M(P,G)$ is the sum of the creation
cost and the loss cost induced by $l_P$. In the example depicted in Figure
\ref{protein-gene-reconciliation}, the creation cost is $2$, while the loss
cost is $3$, and the mutation cost is $5$.\\

\noindent {\bf Homology relations between proteins:} Based on the LCA-reconciliation of $P$ with $G$, we can define the following homology relations between proteins of the set $\mathcal{P}$. Two proteins $x$ and $y$ of $\mathcal{P}$ are called \emph{orthologs} if $l_P(lca_P(\{x,y\})) \neq Creat$, and in this case, we distinguish two types of orthology relationship: $x$ and $y$ are \emph{ortho-orthologs} if $l_P(lca_P(\{x,y\})) = Spec$, and \emph{para-orthologs} otherwise. Note that if $x$ and $y$ are ortho-orthologs (resp. para-orthologs), the genes $g(x)$ and $g(y)$ are orthologs (resp. paralogs). Finally, $x$ and $y$ are \emph{paralogs} if $l_P(lca_P(\{x,y\})) = Creat$. Given a subset $L'$ of $\mathcal{L}(P)$ such that any pair of proteins $(x,y)\in L'^2$ are orthologs, the tree $P|_{L'}$ induced by $L'$ is called a \emph{creation-free subtree} of $P$.\\

\noindent {\bf Problem statements:} 
 Given a protein tree $P$, a gene tree $G$ and a species tree $S$,
 the \emph{double reconciliation cost} of $G$ with $P$ and $S$ is the sum of
 a cost $C_{P\rightarrow G}$ of reconciliation of $P$ with $G$ and a cost
 $C_{G\rightarrow S}$ of reconciliation  of $G$ with $S$.
 
Depending on the costs of reconciliation $C_{P\rightarrow G}$ considered for $P$ with $G$, and $C_{G\rightarrow S}$ considered for $G$ with $S$, nine types of
 double reconciliation cost can be defined. They are denoted by $XY(P,G,S)$
 where $X$ is either $C$ for $C(P,G)$ or $L$ for $L(P,G)$ or $M$ for $M(P,G)$, and $Y$ is either $D$ for $D(G,S)$ or $L$ for $L(G,S)$ or $M$ for $M(G,S)$.
 For example, $CD(P,G,S)$ considers the creation cost for $C_{P\rightarrow G}$ and the duplication cost for $C_{G\rightarrow S}$.

 The definition of the double reconciliation cost naturally leads to the
 definition of our first reconciliation problem that consists in finding
 an optimal gene tree $G$, given a protein tree $P$ and a species tree $S$.
 \\

\noindent \textsc{Minimum Double Reconciliation Gene Tree Problem (MinDRGT$_{XY}$):}\\
\noindent {\bf Input:} A species tree $S$ for $\species$; 
a protein tree $P$ for  $\proteins$; a gene family $\genes$.\\
\noindent {\bf Output:} A gene tree $G$ for $\genes$ that minimizes the double 
reconciliation cost $XY(P,G,S)$.\\

The problem \textsc{MinDRGT} assumes that the protein tree $P$ is known, but in practice, phylogenetic trees on full set of proteins are not available and the application of sequence-based phylogenetic reconstruction methods for constructing protein trees with more than one protein for some genes is likely to lead to incorrect trees as for the reconstruction of single-protein-per-gene trees \cite{hahn2007,rasmussen2011}. However, 
proteins subtrees of $P$ can be obtained by building phylogenetic trees for sets of orthologous protein isoforms \cite{zambelli2010}. Such subtrees can then be combined in order to obtain the full protein tree $P$. One way to combine the orthologous protein subtrees consists in following an approach, successfully used in \cite{lafond2016} for combining a set of gene subtrees into a gene tree. It consists in jointly reconstructing the combined protein tree $P$ and the gene tree $G$ while seeking to minimize the double reconciliation cost of $G$ with $P$ and $S$. We then define a second problem that consists in finding an optimal pair of protein tree $P$ and gene tree $G$, given a species tree $S$ and a set of known subtrees $P_{i, 1 \leq i \leq  k}$ of $P$.

\noindent \textsc{Minimum Double Reconciliation Protein and Gene Tree Problem (MinDRPGT$_{XY}$):}\\
\noindent {\bf Input:} A species tree $S$ for $\species$; 
 a set of proteins $\proteins$, a set of subsets $\proteins_{i, 1 \leq i \leq  k}$ of $\proteins$ such that $\bigcup_{i=1}^{k}{\proteins_i} = \proteins$,
 and a set of protein trees  $P_{i, 1 \leq i \leq  k}$ such that for each $i,1 \leq i \leq  k$, $P_i$ is a tree for $\proteins_i$ and a subtree of the target (real) protein tree.\\
\noindent {\bf Output:} A protein tree $P$ for $\proteins$ such that $\forall i$, $P_{|\proteins_i} = P_i$ and a gene tree $G$ for $\genes = \{g(x) : x \in \proteins \}$ that minimize the double reconciliation cost $XY(P,G,S)$.\\

\section{NP-hardness of \textsc{MinDRGT}}
\label{sec:MinDRGT_hardness}

In this section, we prove the NP-hardness of \textsc{MinDRGT$_{XY}$} for $X = C$ and $Y \in \{D, L, M\}$.
\begin{proposition}
\label{prop:equivalence}
Given a protein tree $P$ on $\proteins$ and a gene tree $G$ on $\genes$ with a protein-species mapping $g$, let $G'$ be a gene tree on $\genes'=\proteins$, and $S'$ a species tree on $\species'=\genes$ with a gene-species mapping $s=g$. The reconciliation costs from $P$ to $G$, and from $G'$ to $S'$
satisfy the following:
(1) $C(P,G)=D(G',S')$ ; (2) $L(P,G)=L(G',S')$ ; (3) $M(P,G)=M(G',S')$.
\end{proposition}

From Proposition \ref{prop:equivalence}, all algorithmic results for the reconciliation
problems between gene and species trees can be directly transferred to the equivalent
reconciliation problems between protein and gene trees. 
In particular, in~\cite{ma2000}, it is shown that, given a gene tree $G$, finding a species tree $S$ minimizing $D(G, S)$ is NP-hard.
To our knowledge the complexity for the same problem with the $L(G, S)$ or $M(G, S)$ costs are still open,
though we believe it is also NP-hard since they do not seem easier to handle than the duplication cost.  
Theorem \ref{thm:NP-hardness} uses these results to
imply the NP-hardness of some versions of \textsc{MinDRGT}.

\begin{theorem}
\label{thm:NP-hardness}
Suppose that the problem of finding a species tree $S$ minimizing the cost $Y'(G, S)$ with a given gene tree $G$ 
is NP-hard for $Y' \in \{D, L, M\}$.  
Let $X = C$ if $Y' = D$, and $X = Y'$ if $Y' \in \{L, M\}$.

Then for any reconciliation cost function $Y \in \{D, L, M\}$ and a given protein tree $P$ and  
species tree $S$, 
the problem of finding a gene tree $G$ minimizing the double-reconciliation cost 
$XY(P, G, S)$ is NP-hard, even if $|\genes| = |\species|$ (Proof given in Appendix).
\end{theorem}

\begin{corollary}
The \textsc{MinDRGT$_{XY}$} problem is NP-hard for $X = C$ and $Y \in \{D, L, M\}$.
\end{corollary}

\section{\textsc{MinDRGT} for the case $\proteins  \Leftrightarrow \genes $}
\label{sec:MinDRGT_P=G}

In Section \ref{sec:MinDRGT_hardness}, we have proved the NP-hardness
of several versions of the \textsc{MinDRGT} problem. In this section, we consider the problem
in a special case where $\proteins \Leftrightarrow \genes$, i.e each gene is the 
image of a single protein by the mapping function $g$. In the remaining of the section, we assume that $\proteins  \Leftrightarrow \genes $ without explicitly
mentioning it. We first study the
subcase where $\proteins \Leftrightarrow \genes \Leftrightarrow \species$, i.e each species contains a single gene of the family. Next, we study the
case where $\proteins \Leftrightarrow \genes$ and develop a heuristic method for it. In the sequel, 
given a protein tree $P$ on $\proteins$, $g(P)$ denotes 
the gene tree for $\genes$ obtained from $P$ by replacing each leaf protein $x\in \proteins$ by the gene $g(x)$ 
(see Figure \ref{fig:cexample3} for example).
Notice that for any cost function $X$, $X(P, g(P)) = 0$.

\subsection{Case where $\proteins \Leftrightarrow \genes \Leftrightarrow \species$.}
\label{sec:MinDRGT_P=G=S}
In this section, we consider the additional restriction that $\genes \Leftrightarrow \species$. For a gene tree $G$ on $\genes$, $s(G)$ denotes the species tree for $\species$ obtained from $G$ by replacing each leaf gene $x\in \genes$ by the species $s(x)$.

One question of interest is whether $g(P)$ is always a solution for \textsc{MinDRGT}$_{XY}$.
In other words, is it the case that for any gene tree $G'$, 
$Y(g(P), S) \leq X(P, G') + Y(G', S)$?  When $X = C$ and $Y = D$, this is true if and only if the  duplication cost satisfies the triangle inequality.  
In~\cite{ma2000}, the authors believed that the duplication cost did have this property, 
but as we show in Figure~\ref{fig:cexample3}, this is not always the case.
In fact, $g(P)$ cannot be assumed to be optimal also for the case $X = Y \in \{L, M\}$.
Thus we get the following remark.

\begin{figure}[!ht]
      \centering
      \includegraphics[width = 0.47\textwidth]{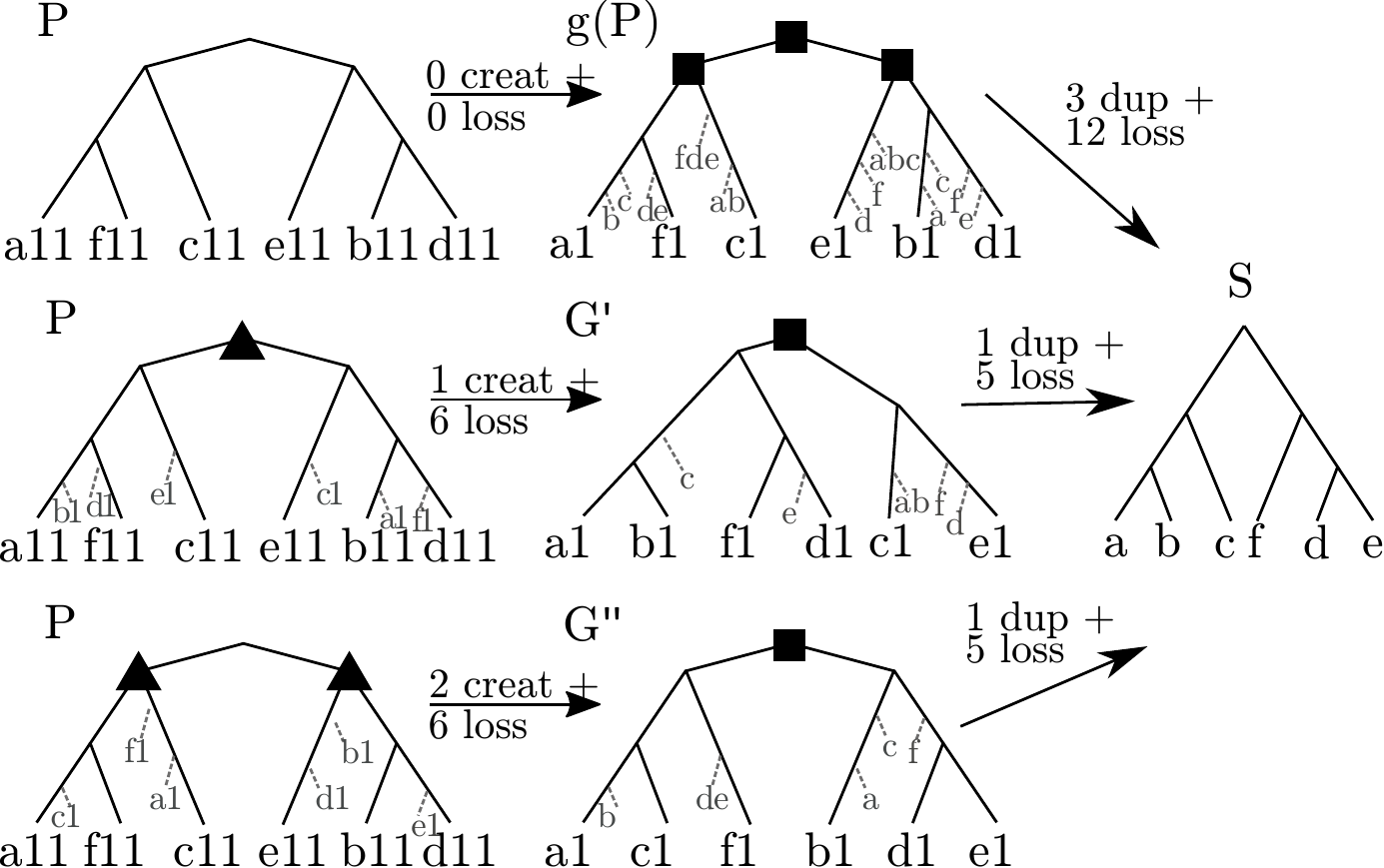}
      \caption{Example of protein tree $P$ on $\proteins = \{a_{11},b_{11},c_{11},d_{11}, e_{11}, f_{11}\}$, species tree $S$ on $\species =\{a,b,c,d,e,f\}$ and gene family $\genes = \{a_1,b_1,c_1,d_1,e_1, f_1\}$, with $g(x_{ij}) = x_i$ and $s(x_i) = x$ for any protein $x_{ij} \in \proteins$, gene $x_i \in \genes$ and species $x \in \species$. The gene tree $G'$ on $\genes$ induces a cost that is strictly lower than the cost induced by the gene tree $g(P)$, for the following double reconciliation costs : CD, CL, CM, LL, LM, MM.
The gene tree $G''$ is the tree that Algorithm 1 would output.}
      \label{fig:cexample3}
    \end{figure}

\begin{remark}Figure \ref{fig:cexample3} illustrates that under the restriction that $\proteins \Leftrightarrow \genes \Leftrightarrow \species$,  in particular the duplication cost, lost cost and the mutation cost do not satisfy the triangle inequality, i.e., there may exists a gene tree $G'$ on $\genes$ such that $D(g(P), S) > C(P, G') + D(G',S)$, $L(g(P),S) > L(g(P),s(G')) + L(G',S)$ and $M(g(P),S) > M(g(P),s(G')) + M(G',S)$. Moreover,
the gene tree $g(P)$ is not a solution for any of \textsc{MinDRGT}$_{CD}$, \textsc{MinDRGT}$_{LL}$, \textsc{MinDRGT}$_{CL}$, \textsc{MinDRGT}$_{CM}$, \textsc{MinDRGT}$_{LL}$, \textsc{MinDRGT}$_{LM}$, and  \textsc{MinDRGT}$_{MM}$.
\end{remark}

\subsection{Case where $\proteins \Leftrightarrow \genes$.}
\label{sec:Heuristic}

In the following, we present a heuristic method
for the \textsc{MinDRGT$_{XY}$} problem, under the restriction 
that $|\proteins | = |\genes|$. The intuition behind the algorithm is based on the idea that we seek for a gene tree
$G_{opt}$ on $\genes$ that decreases the reconciliation cost with $S$, while slightly 
increasing the reconciliation cost with $P$, 
in order to globally decreases the double reconciliation cost with $P$ and $S$.
Let $G=g(P)$. The method consists in building $G_{opt}$ from $G$ by slightly modifying subtrees of 
$G$ that are incongruent with $S$, in order 
to decrease the reconciliation cost with $S$.
The heuristic method makes the following choices:\\
\noindent C1) the subtrees of 
$G$ incongruent with $S$ are those rooted at duplication nodes $x$ with
$l_G(x)\neq l_G(x_l)$ or $l_G(x)\neq l_G(x_r)$. \\
\noindent C2) If $l_G(x) = Dup$ and $l_G(x)\neq l_G(x_l)$ w.l.o.g, we denote by $Mix(G[x])$
the set of trees $G'$ on $\mathcal{L}(G[x])$ that can be obtained by
grafting $G[x_l]$ onto an edge of $G[x_r]$
on which $s(x_l)$ is lost. 
Then, the slight modification applied on $G[x]$ consists in replacing $G[x]$ by a tree $G' \in Mix(G[x])$ that decreases the double
reconciliation cost with $P$ and $S$ by at least $1$.

\noindent {\bf Algorithm 1: Heuristic for \textsc{MinDRGT$_{XY}$}}\\
\noindent Input : Tree $P$ on $\proteins$, Tree $S$ on $\species$, gene set $\genes$, mappings $g,s$.\\
\noindent Output : Tree $G_{opt}$ on $\genes$ such that $G_{opt} = g(P)$ or $XY(P,G_{opt},S) < XY(P,g(P),S) $\\
\noindent 1) $G$ $\leftarrow$ $g(P)$\\
\noindent 2) Compute $l_G$ and let $\baddups = \{x\in \mathcal{V}(G) ~|~ l_G(x)=Dup ~and~ l_G(x)\neq l_G(x_l)~or~ l_G(x)\neq l_G(x_r)\}$\\
\noindent 3) For any node $x\in \baddups$:

a) $u$ $\leftarrow$ the single node $u$ of $P$ s.t.  $g(u) = x$ ;

b) $v$ $\leftarrow$ the single node $v$ of $S$ s.t. $v = s(x)$ ;

c) $G_{opt}[x] \leftarrow argmax_{G' \in Mix(G[x])}{ XY(P[u],G',S[v])}$

d) $\delta(x) \leftarrow Y(G[x],S[v]) - XY(P[u],G_{opt}[x],S[v])$

\noindent 4) Find a subset $\baddups'$  of $\baddups$ s.t. $\forall x \in \baddups'$, $\delta(x) > 0$, and $\forall (x,y) \in \baddups'^2$, $lca(x,y) \neq x$ and $lca(x,y) \neq y$, and $\sum_{x \in \baddups'}{\delta(x)}$  is  maximized.\\
\noindent 5) Build $G_{opt}$ from $G$ by replacing any
subtree $G[x], x\in \baddups'$ by $G_{opt}[x]$.

\noindent {\bf Complexity:} For Step 4 of Algorithm 1, we use a linear-time heuristic greedy algorithm. The time complexity of Algorithm 1 is in $O(n^2)$ where $n= |\genes|$, since $|\mathcal{V}(G)| = O(n)$, $|\baddups| = O(n)$ and $|Mix(G[x])| = O(n)$ for any $x\in \baddups$, and Steps 4
and 5 are realized in linear-time.  

For example, the application of Algorithm 1 on the example
of protein tree $P$ and species tree $S$ depicted in Figure \ref{fig:cexample3} would allow to reconstruct the gene tree $G''$ obtained by moving the subtree of $g(P)$ containing gene $c_1$ onto the branch leading to gene $a_1$, and moving the subtree containing gene $e_1$ onto the branch leading to gene $d_1$. However, the resulting gene tree $G''$ is not as optimal as the gene tree $G'$. Algorithm 1 can be extended in order to allow computing the more optimal gene tree $G'$ by modifying the choices C1 and C2 made by the algorithm: for example, in Step 2 set $\baddups = \{x\in V(G) ~|~ l_G(x)=Dup\}$, and in Step 3.c consider $Mix(G[x]) = \{G' ~|~ G'_{|\mathcal{L}(G[x_l])} = G[x_l] ~ and ~G_{|\mathcal{L}(G[x_r])} = G[x_r]\}$. The resulting algorithm would be an exponential time algorithm because of the exponential size of the sets $Mix(G[x])$.

\vspace{-2mm}

\section{\textsc{MinDRPGT} for maximum creation-free protein subtrees}
\label{sec:MinDRPGT}

\vspace{-2mm}

In this section, we consider the \textsc{MinDRPGT} problem in a special case
where the input subtrees $P_{i,1\leq i \leq k}$ are all the inclusion-wise maximum creation-free protein subtrees of the real protein tree, and we develop an exact algorithm for solving the problem. 

For example, the inclusion-wise maximum creation-free protein subtrees of the labeled protein tree
$P$ depicted in Figure \ref{proteintree} are the subtrees $P_1,P_2,P_3$ of $P$ induced by the subsets of proteins $\proteins_1=\{b_{01},c_{11},a_{31},b_{31},c_{21},c_{31},d_{31}\}$, 
$\proteins_2=\{b_{02},c_{11},a_{31},b_{31},c_{21},c_{31},d_{31}\}$,
and 
$\proteins_3=\{b_{11},a_{21},b_{21},c_{12},a_{31},b_{31},c_{21},c_{31},d_{31}\}$.

Let $P$ be a protein tree for $\proteins$ with a LCA-reconciliation $l_P$,
and $\mathbb{P} = \{P_1, P_2,\ldots,P_k\}$ the set of all the inclusion-wise maximum creation-free protein subtrees of $P$. We define the function $span$ from the
set of protein $\proteins$ to the set $2^{\mathbb{P}}$ of subsets of $\mathbb{P}$ such that,
for any $x\in \proteins$, $span(x)$ is the subset of $\mathbb{P}$ such that $x$ is 
a leaf of any tree in $span(x)$, and $x$ is not a leaf of any tree in $\mathbb{P} - span(x)$. For example, for Figure \ref{proteintree}, $span(b_{01}) = \{P_1\}$, $span(c_{11}) = \{P_1,P_2\}$, $span(b_{11}) = \{P_3\}$, $span(a_{31}) = \{P_1,P_2,P_3\}$.

We define the \emph{span partition} of $\proteins$ according to $\mathbb{P}$ as the partition $\mathbb{P}_{span} = \{S_1, S_2,\ldots,S_m\}$ of $\proteins$
such that for any set $S_u \in \mathbb{P}_{span}$, for any pair of proteins $x,y$ in $S_u$, $span(x)=span(y)$. Note that $\mathbb{P}_{span}$ is unique.  The function $span$ is extended to be defined from $\proteins \cup \mathbb{P}_{span}$ to $2^{\mathbb{P}}$ such that for $S_u \in \mathbb{P}_{span}$, $span(S_u) = span(x)$ for any $x\in S_u$.

For example, for Figure \ref{proteintree},
$\mathbb{P}_{span} = \{S_1 = \{b_{01}\},$ $S_2 = \{b_{02}\},$ $ S_3 = \{b_{11},a_{21},b_{21},c_{12}\},$ 
$S_4 = \{c_{11}\},$ $S_5 = \{a_{31},b_{31},c_{21},c_{31},d_{31}\}\}$: $span(S_1) = \{P_1\}$, $span(S_2) = \{P_2\}$, $span(S_3) = \{P_3\}$, $span(S_4) = \{P_1,P_2\}$  and $span(S_5) = \{P_1,P_2,P_3\}$.

\begin{lemma}
\label{lem:complete-subtrees}
Let $\mathbb{P}$ be the set all inclusion-wise maximum creation-free protein subtrees of a labeled protein tree $P$. Let 
$\mathbb{P}_{span}$ be the span partition of $\proteins$  according to $\mathbb{P}$. 

If $P$ contains at least one protein creation node, then there exist at least a pair of distinct 
sets $S_u,S_v$ in $\mathbb{P}_{span}$ such that
the subtrees of $P$ induced by $S_u$ and $S_v$,  $P_{|S_u}$ and $P_{|S_v}$,
are complete subtrees of $P$, and the subtree of $P$ induced by $S_u \cup S_v$ is also a complete subtree of $P$. In this case:\\
(1) $l_P(lca_P(S_u \cup S_v)) = Creat$ ;\\
(2) For any $t\in \{u,v\}$ and any $P_i \in span(S_t)$, $P_{|S_t} = P_{i|S_t}$ ; \\
(3) $span(S_u) \cap span(S_v) = \emptyset$ ; \\
(4) the following two sets of subtrees are equal: $\{P_{i|\leafset(P_i) - S_u} ~|~ P_i \in span(S_u)\}$ = $\{P_{i|\leafset(P_i) - S_v} ~|~ P_i \in span(S_v)\}$.
\end{lemma}

\begin{proof}
There must exist a node $w$ in $P$ such that $l_P(w) = Creat$
and no node $x\neq w$ in $P[w]$ satisfies
$l_P(x) = Creat$. Then, $S_u = \leafset(P[w_l])$
 and $S_v = \leafset(P[w_r])$.
 \end{proof}

For any set $S_t \in \mathbb{P}_{span}$, $tree(S_t)$ 
denotes the (possibly partial) subtree of $P$
such that, for any $P_i \in span(S_t)$, $tree(S_t) = P_{i|S_t}$.

\noindent {\bf Algorithm 2: for reconstructing $P$ from the set $\mathbb{P}$}\\
\noindent Input : Set $\mathbb{P}$ of all inclusion-wise maximum creation-free protein subtrees of a labeled protein tree $P$ on 
$\proteins$.\\
\noindent Output : Protein tree $P$.\\
\noindent 1) Compute the span partition, $Q \leftarrow \mathbb{P}_{span} = \{S_1, \ldots,S_m\}$ ;  \\
\noindent 2) Set $\{tree(S_u) ~|~ S_u \in \mathbb{P}_{span}\}$ as  subtrees of $P$ ;\\
\noindent 3) While $|Q| > 1$:\\
a) If there exist two distinct sets $S_u,S_v$ in $Q$ such that $span(S_u) \cap span(S_v) = \emptyset$ and $\{P_{i|\leafset(P_i) - S_u} ~|~ P_i \in span(S_u)\}$ = $\{P_{i|\leafset(P_i) - S_v} ~|~ P_i \in span(S_v)\}$, then:

i) Add a node $w$ in $P$ such that $tree(S_u)$, $tree(S_v)$ become the
left and right subtrees of $w$, resulting in a subtree $P'$;

ii) Set $l_P(w) \leftarrow Creat$ and for any $S_t \in Q$, $span(S_t) \leftarrow span(S_t) - span(S_v)$ ;

\noindent b) Otherwise, 

i) Find two distinct sets $S_u,S_v$ in $Q$ such that $P_{|S_u}$ was built at a previous iteration of Step 3;

ii) Graft $P_{|S_u}$ onto $P_{|S_v}$ as the sibling of the node of $P_{|S_v}$ such that the resulting tree $P'$ on $S_u\cup S_v$ is compatible with all subtrees $P_i \in span(S_v)$, i.e. $P'_{|S_t} = P_{i|S_t}$ with $S_t = (S_u\cup S_v)\cap\leafset(P_i)$ ;

\noindent c) Set $S_w \leftarrow S_u\cup S_v$ and $Q \leftarrow Q - \{S_u,S_v\} \cup \{S_w\}$ with 
$span(S_w) \leftarrow span(S_u)$ ;

\begin{theorem}
\label{thm:algo2}
Given the set $\mathbb{P}$ of all inclusion-wise maximum creation-free protein subtrees of a labeled protein tree $P$ 
on $\proteins$, Algorithm 2 reconstruct $P$ and its time complexity is in $O(n^3)$ (Proof given in Appendix).
\end{theorem}
Applying Algorithm 2 on an instance of \textsc{MinDRPGT} such that the input subtrees $P_{i,1\leq i \leq k}$ are all the inclusion-wise maximum creation-free protein subtrees of the real protein tree $P$, allows to reconstruct $P$ with a partial labeling $l_P$ indicating all protein creation nodes. Then, \textsc{MinDRPGT} is
reduced to \textsc{MinDRGT} in the special case where a partial labeling of the input protein tree
is given.

\vspace{-3mm}

\section{Application}\label{sec:Appli}

\vspace{-2.5mm}
 We applied Algorithm 1 for the reconstruction of gene trees using protein 
 trees and gene families of the Ensembl database release 87 \cite{hubbard2002}. 
 Some of the trees were left unchanged by the algorithm.  We call an Ensembl gene tree $G$ `modified' if Algorithm 1, when given $G$, outputs a different tree.  Otherwise we say $G$ is `unmodified'.  The results are summarized in Table \ref{table:results}.
They show that initial gene trees, and particularly large size trees, are predominantly suboptimal in terms of double reconciliation cost. Moreover, modified and unmodified trees have comparable numbers of duplications, but modified trees have significantly higher number of losses, suggesting that gene trees with many losses are susceptible to correction.

\begin{table}[H]
\caption{Results of Algorithm 1 on $10861$ Ensembl gene trees. Samples: (A) $1 \leq n \leq 9$ (7500 trees), (B) $10 \leq n \leq 99$ (2688 trees),  (C) $100 \leq n \leq 199$ (673 trees), where  $n$ is the number of leaves in a tree with the number of trees in each sample in parenthesis 
(1) Number and percentage of modified trees, (2) Average number of duplications / losses in unmodified trees, (3) Average number of duplications / losses in modified trees (before modification), (4) Average value / percentage of double reconciliation cost reduction on modified trees (5) Average running time in ms.}
\centering
\begin{tabular}{|l|l|l|l|l|l|}
\hline
                        & (1) & (2) & (3) & (4) & (5) \\ \hline
(A)    &   \begin{tabular}[c]{@{}l@{}}111 \//\\ 1.48\%\end{tabular}   &  \begin{tabular}[c]{@{}l@{}}0.9\//\\ 7.8\end{tabular}  & \begin{tabular}[c]{@{}l@{}}1.65\//\\ 39.36\end{tabular}    & \begin{tabular}[c]{@{}l@{}}9.40\//\\  18.67\%\end{tabular}    &   0.77  \\ \hline
(B)  & \begin{tabular}[c]{@{}l@{}}1637\//\\ 60.90\%\end{tabular}    &  \begin{tabular}[c]{@{}l@{}}11.15\//\\ 40.67\end{tabular}   & \begin{tabular}[c]{@{}l@{}}11.43\//\\ 115.82\end{tabular}    & \begin{tabular}[c]{@{}l@{}}7.59\//\\ 6.01\%\end{tabular}     &   218.99  \\ \hline
(C) & \begin{tabular}[c]{@{}l@{}}651\//\\ 96.73\%\end{tabular}    & \begin{tabular}[c]{@{}l@{}}41.09\//\\ 167.40\end{tabular}  & \begin{tabular}[c]{@{}l@{}}31.08\//\\307.29\end{tabular}   & \begin{tabular}[c]{@{}l@{}}15.84\//\\ 4.35\%\end{tabular}    &   4724.70  \\ \hline
\end{tabular}
\label{table:results}
\end{table}

\section{Conclusion}
In this work, we have argued the importance of distinguishing 
gene trees from protein trees, and introduced the notion of protein trees into 
the framework of reconciliation.  
We have shown that, just as gene trees are thought of as evolving ``inside'' a species tree,
protein trees evolve ``inside'' a gene tree, leading to two layers of reconciliation.
We also provided evidence that, even if each gene in a given family encodes a single protein, 
the gene phylogeny does not have to be the same as the protein phylogeny, 
and may rather behave like a ``median'' between the protein tree and the species tree in terms of mutation cost.
It remains to evaluate the ability of the developed methods to infer more accurate gene trees on real datasets.

On the algorithmic side, many questions related to the double-reconciliation cost 
 deserve further investigation.  For instance, what is the complexity of finding an optimal gene tree 
in the case that $\proteins \Leftrightarrow \genes \Leftrightarrow \species$?  Also, given that the general 
\textsc{MinDRGT} problem is NP-hard, can the optimal gene tree $G$ be approximated
within some constant factor?  Or is the problem fixed-parameter tractable with respect to some interesting 
parameter, e.g. the number of apparent creations in the protein tree, or the maximum number of proteins per gene?
As for the \textsc{MinDRPGT} problem, it remains to explore how the partially labeled protein trees can be used to infer the gene tree.  Moreover, we have studied an ideal case where all maximum creation-free 
protein subtrees could be inferred perfectly.  Future work should consider relaxing this assumption
by allowing the input subtrees to have missing or superfluous leaves, or to contain errors.

\begin{spacing}{0.9}
\bibliographystyle{plain}
\bibliography{biblio}
\end{spacing}

\end{document}